\documentclass{amsart}

\usepackage{mathptmx}      
\usepackage{amsmath}
\usepackage{amsfonts}
\usepackage{amssymb}
\usepackage{amsmath}
\usepackage{color}
\usepackage[numbers,square]{natbib}
\usepackage{graphicx}
\usepackage[letterpaper]{geometry}

\renewcommand{\S}{S^*} 
\newcommand{\T}{T^*}
\newcommand{\reg}{\mathop{\textup {Reg}}}
\newcommand{\unreg}{\mathop{\textup {Unreg}}}

\renewcommand{\bar}{\overline}

\newcommand{\starves}{\succ}

\definecolor{Edgar}{rgb}{0,0.08,0.6}
\definecolor{Comment}{rgb}{0.5,0,0.5}

\newcommand{\comment}[1]{}
\renewcommand{\comment}[1]{{\color{Comment} #1}}

\newtheorem{theorem}{Theorem}
\newtheorem{proposition}{Proposition}
\newtheorem{remark}{Remark}
\newtheorem{definition}{Definition}
\newtheorem{corollary}{Corollary}
\newtheorem{lemma}{Lemma}

\begin{document}

\title{A model of host response to a multi-stage pathogen}

\author{Edgar Delgado-Eckert}
\address{Departent of Biosystems Science and Engineering, Swiss Federal Institute of
Technology Zurich (ETH Z\"urich), Basel, Switzerland.}
      
\author{Michael Shapiro}
\address{Department of Pathology, Tufts University, Boston, MA 02111}
\email{michael.shapiro@tufts.edu}
\thanks{The second author wishes to acknowledge support from  NIH
  Grant K25 AI079404-02.}

\date{\today}

\begin{abstract}
We model the immune surveillance of a pathogen which passes through $n$
immunologically distinct stages. The biological parameters of this system
induce a partial order on the stages, and this, in turn, determines which
stages will be subject to immune regulation. This corresponds to the
system's unique asymptotically stable fixed point.
\end{abstract}

\maketitle

\section{Introduction}

Pathogens that traverse different stages during their life cycle or during
an infection process have been studied since the late nineteenth century.
The most prominent genus is \textit{Plasmodium}, causer of Malaria, a
pathology first described by Hippocrates in the fourth century B.C., which
is still a major public health problem in many continents \cite%
{MalariaReview}. Other important examples are Trypanosoma \cite%
{ChagasReview} and Epstein-Barr virus (EBV) which is a member of the family
of herpes viruses. Indeed, it seems likely that this entire family uses
multiple stages \cite{ThorleyLawson2008195}. Our focus is on Epstein-Barr
virus, which is known to cycle through at least four different stages during
infection within the human body \cite{EBVreview}. (For a quick sketch of
EBV biology, see Section~\ref{sec:discussion}). One remarkable
characteristic of infections with many of such pathogens is life-long
persistent infection \cite{[17]}, \cite{ThorleyLawson2008195}, \cite{[18]}, 
\cite{MalariaReview}, \cite{ChagasReview}.

Our main goal in this article is to study the properties of the immune
response to such a pathogen using mathematical modeling. Mathematical
approaches to studying host-pathogen interactions have steadily increased in
the last 4 decades. For entry into the corresponding body of literature, we
recommend \cite{NowakMay}, \cite{Perelson1999} and \cite%
{KillerCellDynamics}. What to include in such a mathematical model depends
not only on the host and pathogen in question, but also the taste and
purpose of the modeler. Thus models vary in scope and complexity. A simple
model might include uninfected tissue, infected tissue and an immune
response. More complex models can involve such factors as the details of
T-cell activation, competing immune responses to multiple epitopes and the
mutation, selection and consequent evolution of the pathogen. Much of this
literature is driven by the urgent need to understand the host-pathogen
dynamics of HIV infection. (A Pubmed search on the terms \textquotedblleft
HIV\textquotedblright\ and \textquotedblleft mathematical
model\textquotedblright\ produces close to 400 hits.) Our model is motivated
by the broad outlines of EBV biology which it both simplifies and
generalizes.

In the mathematical model we propose herein, the pathogen traverses a cycle
of immunologically distinct stages during infection of the host. On the
pathogen side, our model includes the natural proliferation or decay of the
pathogen population at each stage, the rate at which each stage is lost to
produce (or become) the next stage and any gain factor involved in this
process. (For example, one lytically infected cell may produce $\sim 10^{4}$
free virions.) On the host side, it includes a distinct immune response to
each stage, proliferation of this response due to encounter with the
pathogen, the ability of this response to kill or disable the pathogen, and
the decay of the response in the absence of the pathogen. Omitted from this
model is the supply of uninfected tissue which supports the cycle of
infected stages. This omission is only likely to be important when this
supply is a limiting factor. This could happen when the immune response is
weak or non-existent, as when a naive host is first exposed or in the case
of an immunosuppressed host. As a consequence, this model may be more
realistic in describing the long-term behavior of persistent infection when
this is not tissue-limited.

We will see that when the host and pathogen stand in this relation to each
other, there is a unique stable accommodation between them. This appears in
the model as the unique biologically meaningful asymptotically stable fixed
point. In particular, it determines exactly which stages will fall under
immune regulation and determines the levels of pathogen population and host
response at which this mutual accommodation occurs.

In Section \ref{sec:formulation} we will present the model's equations and
the range of parameter values under consideration. We will start by
considering parameter values such that the pathogen must traverse all stages
in order to establish infection. In particular, no stage is capable of
proliferating and establishing infection independent of the others. In
addition, we will restrict our attention to \emph{generic} values of the
parameters and will exclude certain parameter sets of measure zero.

The central topic of this article will be the analysis of the steady state
behavior of our model. To this end, we start by calculating necessary and
sufficient conditions for the pathogen to establish infection. We view this
both in terms of the basic reproductive constant, $R_{0}$, and in terms of
the linear stability analysis of the uninfected equilibrium at which all
stages of the pathogen as well as their corresponding immune responses
become extinct.

We continue the steady state analysis by determining further fixed points of
the model's equations. As we will see, for generic parameter values, this
system has exactly $2^{n}$ fixed points, (where $n$ is the number of
stages), depending on which stages are regulated by the immune response. In
general, not all of these are biologically relevant since a given pattern of
regulated and unregulated stages can give negative population levels at
steady state. We give a criterion for detecting the biologically meaningful
ones.

If the pathogen is able to establish infection, ($R_{0}>1$), the parameters
induce a partial order on the pathogen's stages. We say stage $j$ starves
stage $k$ (in symbol $j\succ k$) if immune regulation at stage $j$ deprives
stage $k$ of sufficient population to support immune regulation. A stage $k$
is called \emph{starvable }if there is another stage $j$ such that $j\succ
k. $ If no such $j$ exists, $k$ is called \emph{unstarvable}. One of our
main results is the fact that, generically, the system has a unique
asymptotically stable fixed point, namely, the one at which all unstarvable
stages are regulated and all starvable stages are unregulated.

The portrait we have just drawn holds for generic values of the parameter
set. It turns out that there are highly non-generic parameter sets for which
the system becomes a Lotka-Volterra predator-prey system. While
mathematically interesting, this is unlikely to occur in nature. The
condition that the parameters must obey has co-dimension $n-1$. In
particular, such parameters form a set of measure zero.

Finally, we show how to generalize the previous results to the case where
the pathogen has individual stages which are capable of establishing
infection independently of the remaining stages.

This article is organized as follows: In Section~\ref{sec:formulation}, we
introduce the model and discuss its parameters. In Section~\ref%
{sec:FixedPoints}, we examine the basic properties of the system's fixed
points. In Sections \ref{sec:R0},~\ref{sec:PartialOrder} and~\ref%
{sec:stability} we embark on the detailed analytical study of the model's
equilibrium behavior. One of the main tools used is the partial order on
stages, defined in Section~\ref{sec:PartialOrder}. In Section \ref%
{sec:Lotka-Volterra}, we present the highly non-generic Lotka-Volterra
scenario. Finally, in Section~\ref{sec:selfEstablishingStages} we consider
self-establishing stages. We close in Section~\ref{sec:discussion} with
discussing some of the model assumptions and consequences for the
characteristics of the immune response to a multi-stage pathogen. To provide
a more concrete contrast between the model and reality, we focus on EBV
infection biology.

\section{Formulation of the model}

\label{sec:formulation}

We model the interaction between a host and a pathogen which traverses $n\in%
\mathbb{N}$, $n>1$ different stages during its life-cycle. We treat these
stages as immunologically distinct, that is, the host produces a separate
immune response to each of them. Our model is given by 
\begin{align*}
\dot{S_{j}}&=F_{j}(S,T)=r_{j-1}f_{j-1}S_{j-1}
-a_{j}S_{j}-f_{j}S_{j}-p_{j}S_{j}T_{j}  \tag{*}  \label{Eq.BiologicalModel}
\\
\dot{T_{j}}& =G_{j}(S,T)=c_{j}S_{j}T_{j}-bT_{j}.
\end{align*}

The subscripts are taken modulo $n$. We use $S_j$ and $T_j$ for the pathogen
population and immune response\footnote{%
The notation $T_j$ is motivated by the T-cell response. However, it may
equally well refer to humoral response.} at each stage. These are assumed
non-negative. The parameters represent the following processes:

\begin{itemize}
\item $a_{j}$ is the decay rate of stage $S_{j}$. If $a_{j}$ is negative,
this state proliferates.

\item $f_{j}$ is the rate at which stage $S_{j}$ is lost to become (or
produce) stage $S_{j+1}$.

\item $r_{j}$ is an amplification factor in the process by which stage $%
S_{j} $ becomes (or produces) stage $S_{j+1}$. For example, the loss of one
lytically infected cell may produce $r_{j}\cong 10^{4}$ free virus.

\item $p_{j}$ represents the efficacy of the immune response $T_{j}$ in
killing infected stage $S_{j}$.

\item $c_{j}$ is the antigenicity of stage $S_{j}$, i.e., its efficacy in
inducing proliferation of immune response $T_{j}$.

\item $b$ is the natural death rate of the response $T_{j}$. We assume it is
the same for all stages.
\end{itemize}

We will refer collectively to a choice of values for these parameters as $%
\pi $. We will assume that with the possible exception of $a_{j}$, all of
these are positive. In Sections~\ref{sec:FixedPoints} through~\ref%
{sec:Lotka-Volterra}, we will assume that $a_{j}+f_{j}>0$. Thus, while some
stages may proliferate, none does this as fast as it is consumed in
producing the next stage. In Section~\ref{sec:selfEstablishingStages}, we
will lift this assumption and generalize our results to the case in which
there are self-establishing stages. We will also assume that the values of
these parameters are \emph{generic}. Accordingly, at several places we will
disregard certain sets of parameters which have measure zero. (See
Equations~ (\ref{Eq.FirstGenericityAssumption}) and (\ref%
{Eq.SecondGenericityAssumptionVersion1}).)

The model's equations can be simplified by using the following change of
coordinates 
\begin{align*}
\bar{S}_{j}& =c_{j}S_{j} \\
\bar{T}_{j}& =p_{j}T_{j}
\end{align*}
which produces the equations 
\begin{align*}
\dot{\bar{S}}_{j}& =\bar{F}_{j}(\bar{S},\bar{T})=\bar{r}_{j-1}f_{j-1}\bar{S }%
_{j-1}-a_{j}\bar{S}_{j}-f_{j}\bar{S}_{j}-\bar{S}_{j}\bar{T}_{j}  \notag \\
{\dot{\bar{T}}_{j}}& =\bar{G}_{j}(\bar{S},\bar{T})=\bar{S}_{j}\bar{T}_{j}-b 
\bar{T}_{j}  \tag{**}  \label{Eq.BasicModel} \\
\bar{r}_{j}& =\frac{c_{j+1}}{c_{j}}r_{j}  \notag
\end{align*}
We will henceforth assume our equations are in this form and omit the bars.
When we refer to a fixed point, we will mean a fixed point of system~(\ref%
{Eq.BasicModel}). The notations $(\S _{0},\dots ,\S _{n-1},T_{0}^{\ast
},\dots ,T_{n-1}^{\ast })$ or $(\S ,T^{\ast })$ will refer to such a fixed
point. We will say that a point is \emph{\ biologically meaningful} if its
components are non-negative.

We will adopt the following notational conventions. Sets such $[j,k]$ and $%
[j,k)$ are to be taken cyclically. That is to say, if $j < k$, then $[j,k] =
\{j,\dots,k\}$, while if $j > k$, $[j,k] = \{j,\dots,n-1,0,\dots,k\}$. We
take $[j,j)$ to be the empty product so that $[j,j)=1$. We abuse notation by
taking $[0,n) = \{0,\dots,n-1\}$.

\section{First properties of the fixed points}

\label{sec:FixedPoints}The fixed points are defined by the system of
polynomial equations 
\begin{eqnarray}
\dot{S}_{j} &=&F_{j}(S^{\ast },T^{\ast })=r_{j-1}f_{j-1}S_{j-1}^{\ast
}-S_{j}^{\ast }(a_{j}+f_{j}+T_{j}^{\ast })=0  \label{Eq.FixedPointEquations1}
\\
\dot{T}_{j} &=&G_{j}(S^{\ast },T^{\ast })=(S_{j}^{\ast }-b)T_{j}^{\ast }=0
\label{Eq.FixedPointEquations2}
\end{eqnarray}
where $j=0,...,n-1$ and indices outside the interval $[0,n-1]$ are
understood modulo $n.$

In this section, we give formulas that determine the population values at
fixed points of system~(\ref{Eq.BasicModel}). We will see that if stage $k$
is unregulated, $\S _j$ is determined as a ``follow-on'' population from the
previous stage. This gives rise to ``follow-on factors'' (defined below)
which will play a key role in this paper. This, in turn, will lead us to the
notion of generic parameter values, and we will see that for generic
parameter values, system~(\ref{Eq.BasicModel}) has exactly $2^n$ fixed
points.

This section is organized as a series of observations. These follow from
easy computations which we omit when they would clutter rather than clarify.

We start by defining these follow-on factors 
\begin{align*}
M_{j}& :=\frac{r_{j}f_{j}}{a_{j+1}+f_{j+1}} \\
M_{jk}& :=\prod_{\ell \in \lbrack j,k)}M_{\ell } \\
R_{0}& :=\prod_{j\in \lbrack 0,n)}M_{j}.
\end{align*}%
which we will use throughout this article. We will take $M_{jj}$ to be the
empty product. Given a fixed point $(\S ,T^{\ast })$, if $T_{j}^{\ast }\neq
0 $, we will say that $j$ is \emph{regulated}. Otherwise $j$ is \emph{\
unregulated}. We define 
\begin{align*}
\mathop{\textup {Reg}}(\S ,T^{\ast })& =\{j\mid T_{j}^{\ast }\neq 0\} \\
\mathop{\textup {Unreg}}(\S ,T^{\ast })& =\{j\mid T_{j}^{\ast }=0\}
\end{align*}

\begin{enumerate}
\item[\textbf{1)}] ~~ \textit{If there is a $j$ such that $\S _{j}=0$ then $(%
\S ,T^{\ast })=(\vec{0},\vec{0})$. } By (\ref{Eq.FixedPointEquations1}) if $%
\S _{j}=0$, then $\S _{j-1}=0$. By (\ref{Eq.FixedPointEquations2}) if $\S %
_{j}=0$, then $T_{j}^{\ast }=0$. Continuing in this way $\S _{j}=T_{j}^{\ast
}=0$ for all $j\in \lbrack 0,n)$.

\item[\textbf{2)}] ~~ \textit{If $j\in \mathop{\textup {Reg}}(\S ,T^{\ast
}), $ then $\S _{j}=b$.}

\item[\textbf{3)}] ~~ \textit{If $j\in \mathop{\textup {Unreg}}(\S ,T^{\ast
}),$ then $T_{j}^{\ast }=0$.}

\item[\textbf{4)}] ~~ \textit{If $j+1 \in \mathop{\textup {Unreg}}(\S .T^*)$%
, then $\S _{j+1} = \S _j M_j$.} This follows quickly from (\ref%
{Eq.FixedPointEquations1}).

\item[\textbf{5)}] ~~ \textit{If $[j+1,k]\subset \mathop{\textup {Unreg}}(\S %
,T^{\ast }),$ then $\S _{k}=\S _{j}M_{jk}$. } This follows by induction on
the previous observation.

\item[\textbf{6)}] ~~ \textit{If $j\in \mathop{\textup {Reg}}(\S ,T^{\ast })$
and $[j+1,k]\subset \mathop{\textup {Unreg}}(\S ,T^{\ast }),$ then $\S %
_{k}=bM_{jk}$.}

\item[\textbf{7)}] ~~ \textit{If $R_{0}\neq 1$, and $(\S ,T^{\ast })\neq (%
\vec{0},\vec{0}),$ then $\mathop{\textup {Reg}}(\S ,T^{\ast })\neq \emptyset 
$.} To see this, notice that by our first observation, we must have $\S %
_{j}\neq 0$ for all $j\in \lbrack 0,n)$. But if $\mathop{\textup {Reg}}(\S %
,T^{\ast })=\emptyset $, we must have $\S _{0}=R_{0}\S _{0}$.
\end{enumerate}

Before proceeding, two points are worth considering. The first is that we
have just introduced our first condition for a generic parameter set $\pi$,
that is, 
\begin{equation}
R_{0}= \prod_{j\in [0,n)} M_j \neq 1  \label{Eq.FirstGenericityAssumption}
\end{equation}
The second is that the $S$ population at an unregulated stage depends on
follow-on constants and the previous regulated stage. Observation 7
guarantees that at a non-trivial fixed point, there always is such a stage.
Accordingly, given $k$, we will define $h_k$ to be this stage, that is, $h_k$
is the unique stage such that $h_k\in\mathop{\textup {Reg}}(\S ,T^*)$ and $%
[h_k+1,k) \subset \mathop{\textup {Unreg}}(\S ,T^*)$. As we will see, this
notation is also useful in the case $k\in \mathop{\textup {Reg}}(\S ,T^*)$.

\begin{enumerate}
\item[\textbf{8)}] ~~ \textit{If $j\in\mathop{\textup {Reg}}(\S ,T^*)$, then 
$T^*_j = \frac{r_{j-1}f_{j-1}}{b}\S _{j-1} - (a_j+f_j)$.} This follows from (%
\ref{Eq.FixedPointEquations1}).

\item[\textbf{9)}] ~~ \textit{If $j\in\mathop{\textup {Reg}}(\S ,T^*)$, then 
$T^*_j = r_{j-1}f_{j-1} M_{{h_j}{j-1}} - (a_j+f_j)$. } This follows from the
fact that $\S _{j-1} = \S _{h_j}M_{{h_j}{j-1}}$ and $\S _{h_j} = b$.

\item[\textbf{10\label{item:starves})}] ~~ \textit{$T^*_j > 0$ if and only
if $M_{{h_j} j} > 1$.}
\end{enumerate}

This latter will be important in Section~\ref{sec:PartialOrder} and leads to
our second genericity requirement. We take 
\begin{equation}
M_{jk} \ne 1 \text{ for $j\ne k$.}
\label{Eq.SecondGenericityAssumptionVersion1}
\end{equation}
This may seem overly broad since we have only used $M_{j k } \ne 1$ for $j,
k \in \mathop{\textup {Reg}}(\S ,T^*)$, with $j=h_k$. However, we wish to
avoid making the genericity of $\pi$ depend on the fixed point in question.
The sets in question are all of co-dimension 1 and hence have measure 0.

We have now arrived at the place where all of the populations are determined
by the parameters and the pattern of regulation. We record these as follows:
 
\begin{equation}
\S_j=b \text{ and } \T_j = r_{j-1}f_{j-1} M_{h_j j-1} - (a_j+f_j) \text{ } \forall \text{ }
j\in\reg(\S,\T) 
\\
\label{Eq.MultipleStepTvalues}
\end{equation}%
\begin{equation}
\S_j = b M_{h_j j} \text{ and } \T_j=0 \text{ } \forall \text{ }
j\in\unreg(\S,\T) 
\label{Eq.MultipleStepSvalues}
\end{equation}%

We are now prepared to show that for generic parameters, (\ref{Eq.BasicModel}%
) has exactly $2^n$ fixed points. We will show that for an arbitrary set $R
\subset [0,n)$, there is exactly one fixed point $(\S ,T^*)$ with $%
\mathop{\textup {Reg}}(\S ,T^*) = R$. If $R=\emptyset$, we have seen that $(%
\S ,T^*) = (\vec{0},\vec{0})$. Under the assumption that $R\ne \emptyset$, $%
\S _j$ and $T^*_j$ are determined for $j\in R$ by (\ref{Eq.MultipleStepTvalues}).
Observe that if $\pi$ is generic, then the $T$-values
determined by (\ref{Eq.MultipleStepTvalues}) are non-zero, so that $\mathop{\textup {Reg}}(\S ,T^*) = R$
as required. Since the population
values for $j\notin R$ are now determined by (\ref{Eq.MultipleStepSvalues}),
we see that there is exactly one fixed point for each $R\subset [0,n)$ as
required.

In general, not all of these non-vanishing fixed points (defined by (\ref%
{Eq.MultipleStepTvalues}) and (\ref{Eq.MultipleStepSvalues}) under the
genericity assumptions) are biologically relevant, since a given pattern of
regulated and unregulated stages can give $T_{j}<0$ for one or more values
of $j$. We will call non-vanishing fixed points (i.e., different from the
uninfected equilibrium) \emph{infected equilibria}. Infected equilibria with
the property $T_{j}^{\ast }\geq 0$ $\forall $ $j=0,...,n-1$ ($S_{j}^{\ast
}>0 $ $\forall $ $j=0,...,n-1$ follows from the assumed positivity of
parameters) are called \emph{biologically meaningful infected equilibria}.

\section{Establishing infection: Stability of the uninfected equilibrium}

\label{sec:R0}

We start by deriving the expression for the basic reproductive number $R_{0}$
already revealed above (see (\ref{Eq.FirstGenericityAssumption})). Our first
genericity assumption (\ref{Eq.FirstGenericityAssumption}) allows us to
eliminate the case $R_{0}=1$. We will then show that the remaining dichotomy 
$R_{0}<1$ or $R_{0}>1$ determines whether the pathogen is able to establish
infection. We give two proofs here. The first uses the interpretation of the
parameters as describing the life-cycle of the pathogen. The second is based
on the analysis of eigenvalues of the Jacobian at the uninfected fixed
point. These two approaches correspond roughly to $R_{0}$ and $r_{0}$. For
an excellent discussion of these, see \cite{heffernan}.

Given $\pi $ and an arbitrary biologically meaningful point $%
(S,T)=(S_{0},\dots ,S_{n-1},T_{0},\dots ,T_{n-1})$, we define 
\begin{align*}
\widetilde{M}_{j}& =\frac{r_{j}f_{j}}{a_{j+1}+f_{j+1}+T_{j+1}} \\
\widetilde{R}& =\prod_{j=0}^{n-1}\widetilde{M}_{j}
\end{align*}
Note that at a fixed point $(S^{\ast },T^{\ast })$, $\widetilde{M}_{j}=M_{j}$
for $j\in \mathop{\textup {Unreg}}(S^{\ast },T^{\ast })$.

\begin{proposition}
\label{prop:R0} ~

\begin{enumerate}
\item The basic reproductive number of the system is $R_{0} $.

\item The reproductive number at an arbitrary biologically meaningful point $%
(S,T)$ is $\widetilde{R}$.

\item The overall reproductive number of the pathogen in the presence of any
(non-negative) immune response is less than its basic reproductive number in
the absence of an immune response.

\item If $R_{0}<1$, the pathogen fails to establish infection and the
uninfected fixed point $(\vec{0},\vec{0})$ is a global attractor.

\item At an infected fixed point, the reproductive number is exactly $1 $.
Thus, at an infected fixed point we have $\prod_{j\in \lbrack 0,n)} 
\widetilde{M}_{j}=1$.
\end{enumerate}
\end{proposition}

\begin{proof}
To see that $R_{0}$ is the basic reproductive constant of this system,
consider the behavior of stage $S_{0}$ when the naive host encounters one
unit of this stage. The average lifespan at this stage is $1/(a_{0}+f_{0})$.
During the course of that lifespan, it produces $r_{0}f_{0}/(a_{0}+f_{0})$
units of stage $S_{1}$. This, in turn produces $%
(r_{1}f_{1}/(a_{1}+f_{1}))(r_{0}f_{0}/(a_{0}+f_{0}))$ units of stage $S_{2}$
. Continuing in this way, we see that the original introduction of a unit of 
$S_{0}$ results in $\prod_{j=0}^{n-1}(r_{j}f_{j}/(a_{j}+f_{j}))=
\prod_{j=0}^{n-1}M_{j}=R_{0}$ units of $S_{0}$.

To see that the basic reproductive number at an arbitrary point $(S,T)$ is $%
\prod_{j\in \lbrack 0,n)}\widetilde{M}_{j}$, observe that in the presence of
an immune response, the average lifespan of stage $S_{j}$ is $%
1/(a_{j}+f_{j}+T_{j})$. The result now follows as before. The third and
fourth statements now follow immediately, while the final statement follows
from observing that at a fixed point the pathogen exactly reproduces itself
and therefore has basic reproductive number 1. 
\end{proof}

We now examine $R_0$ from a slightly different perspective.

\begin{proposition}
\label{prop:hereBePartials} If $R_{0}<1,$ the eigenvalues of the Jacobian
matrix of system (\ref{Eq.BasicModel}) evaluated at $(S,T)=(\vec{0},\vec{0})$
have negative real part and, consequently, the uninfected fixed point is
asymptotically stable. On the other hand, if $R_{0}>1,$ at least one
eigenvalue of the Jacobian has positive real part and thus, the uninfected
equilibrium is unstable.
\end{proposition}

\begin{proof}
The characteristic polynomial $P$ of $J(S,T)$ is given by the following
expression. 
\begin{equation*}
P(\lambda )=(-1)^{n}\prod\limits_{j=0}^{n-1}\left( (S_{j}-\sigma
_{j}-\lambda )(b+\lambda )-S_{j}(-(a_{j}+f_{j})+b)\right)
-\prod\limits_{j=0}^{n-1}\left( \rho _{j}(b+\lambda -S_{j+1})\right)
\end{equation*}
where we use the notation $\rho _{j}=r_{j}f_{j}$ and $\sigma
_{j}=a_{j}+f_{j}+T_{j}.$ (See Theorem~\ref{thm:charicteristicPolynomial} of
the Appendix for the details on how these expressions were derived.)
Evaluating these at $(\S ,T^{\ast })=(\vec{0},\vec{0})$ gives 
\begin{equation*}
P(\lambda )=(-1)^{n}(b+\lambda )^{n}\left(
\prod\limits_{j=0}^{n-1}(a_{j}+f_{j}+\lambda
)-\prod\limits_{k=0}^{n-1}r_{k}f_{k}\right) .
\end{equation*}
Thus, the eigenvalues of $J(\vec{0},\vec{0})$ are given by the $n$-fold root 
$\lambda =-b$ and the zeros of the polynomial $R(\lambda
):=\prod\limits_{j=0}^{n-1}(a_{j}+f_{j}+\lambda
)-\prod\limits_{k=0}^{n-1}r_{k}f_{k}.$ By our assumptions, all coefficients
of $R$ are positive except for $\prod\limits_{j=0}^{n-1}(a_{j}+f_{j})-\prod
\limits_{k=0}^{n-1}r_{k}f_{k},$ which is positive if and only if 
\begin{equation*}
R_{0}=\frac{\prod\limits_{k=0}^{n-1}r_{k}f_{k}}{\prod
\limits_{j=0}^{n-1}(a_{j}+f_{j})}<1
\end{equation*}
By Routh's criterion (\cite{OnRouthsCriterion}), if $R_{0}>1,$ at least
one root of $R$ (and therefore of the characteristic polynomial $P$) has
positive real part. Therefore, the equilibrium point $(S,T)=(\vec{0},\vec{0}
) $ becomes unstable for $R_{0}>1.$

We now wish to see that if $R_{0}<1,$ all roots of $P$ have strictly
negative real part implying that the equilibrium point $(S,T)=(\vec{0},\vec{0 })$ is (locally) asymptotically stable.

To see this, consider the polynomials 
\begin{equation*}
Q(\lambda ):=\prod\limits_{j=0}^{n-1}(a_{j}+f_{j}+\lambda )
\end{equation*}
and 
\begin{equation*}
E(\lambda ):=-\prod\limits_{j=0}^{n-1}r_{j}f_{j}.
\end{equation*}
Thus, $R=Q+E,$ and we can interpret $R$ as a perturbation of $Q$ by $E.$ By
our assumptions, the real parts of all roots of $Q$ are strictly negative,
in other words, $Q$ is strictly Hurwitz\footnote{%
A real polynomial is called \emph{strictly Hurwitz} if the real part of each
of its (complex) roots is strictly negative.}. Here the question arises as
to how big the perturbation $E$ can be, so that $R$ remains strictly
Hurwitz. By Corollary 4 in \cite{Lin'sPaperOnRobustness}, this will occur
if\footnote{%
We will see that such a maximum exists.} 
\begin{equation*}
\max_{\omega \geq 0}\left\vert \frac{E(i\omega )}{Q(i\omega )}\right\vert <1
\end{equation*}
where $i$ is the imaginary unit. We have 
\begin{align*}
\left\vert E(i\omega )/Q(i\omega )\right\vert & =\left\vert \left(
-\prod\limits_{j=0}^{n-1}r_{j}f_{j}\right) /\left(
\prod\limits_{j=0}^{n-1}(a_{j}+f_{j}+i\omega )\right) \right\vert \\
& =\sqrt{\left( \prod\limits_{j=0}^{n-1}r_{j}f_{j}\right) ^{2}/\left(
\prod\limits_{j=0}^{n-1}(a_{j}+f_{j}+i\omega )(a_{j}+f_{j}-i\omega )\right) }
\\
& =\sqrt{\left( \prod\limits_{j=0}^{n-1}r_{j}f_{j}\right) ^{2}/\left(
\prod\limits_{j=0}^{n-1}((a_{j}+f_{j})^{2}+\omega ^{2})\right) }
\end{align*}
It is now clear that this achieves its maximum, namely $R_{0}$, at $\omega
=0 $. In particular, $R_{0}<1$ implies that all eigenvalues of the Jacobian
have negative real part and the uninfected fixed point is asymptotically
stable. 
\end{proof}

\section{Starvation: A partial order on the stages}

\label{sec:PartialOrder}

In this section, we introduce a relation $j\succ k$ which we read as $j$ 
\emph{starves} $k$. This relation depends only on $\pi $. We will show that
this relation limits which patterns of regulation can appear in a
biologically meaningful fixed point, and will show that $j\succ k$ is a
strict partial order.

\begin{definition}
\label{def:starves}We will say stage $j$ \emph{starves} stage $k$ and write $%
j\succ k$ iff $M_{jk}<1.$
\end{definition}

The following proposition explains our use of the term starves.

\begin{proposition}
\label{prop:starves}Let the set of parameters $\pi $ be generic.
Furthermore, let $(S^{\ast },T^{\ast })$ be a biologically meaningful
infected equilibrium (in particular, $\textup{Reg}(S^{\ast },T^{\ast })\neq
\emptyset $) and let $j,k\in \lbrack 0,n)$ such that $j\neq k.$ If $j\succ k$
, then $j\in \mathop{\textup {Reg}}(S^{\ast },T^{\ast })$ implies $k\in %
\mathop{\textup {Unreg}}(S^{\ast },T^{\ast })$.
\end{proposition}

\begin{proof}
Suppose $j\in \mathop{\textup {Reg}}(S^{\ast },T^{\ast })$ and assume to the
contrary that $k\in \mathop{\textup {Reg}}(S^{\ast },T^{\ast })$. Let $m\in {%
\ \mathbb{N}}_{0}$ be the number of regulated stages between $j$ and $k$.
That is, suppose $[j+1,k)\cap \mathop{\textup
    {Reg}}(S^{\ast },T^{\ast })=\{j_{1},\dots ,j_{m}\}$ and suppose  these
are ordered so that (cyclically) $j=j_{0}\leq j_{1}\leq  j_{2}\leq ...\leq
j_{m}<j_{m+1}=k$. Since $(S^{\ast },T^{\ast })$  is biologically meaningful
and $\pi $ is generic, $T_{j}^{\ast }>0$  for $j\in \mathop{\textup {Reg}}%
(S^{\ast },T^{\ast })$. By  observation~\ref{item:starves} of Section~\ref%
{sec:FixedPoints}, we  have $M_{j_{s}j_{s+1}}>1$ for $s=0,\dots ,m$. But
this gives 
\begin{equation*}
M_{jk}=\prod_{s=0}^{m}M_{j_{s}j_{s+1}}>1,
\end{equation*}
contradicting the assumption that $j\succ k$. 
\end{proof}

This result can be biologically interpreted as follows: If $j\succ k$,
regulation at stage $S_{j}$ starves stage $S_{k}$ of sufficient population
to support regulation.

\begin{proposition}
Let $\pi $ be such that $R_{0}>1.$ Then $\succ $ is a strict partial order.
\end{proposition}

\begin{proof}
We must show that $\succ $ is transitive, anti-reflexive and asymmetric.
Suppose $j\succ k$ and $k\succ \ell $. If $k\in \lbrack j,\ell ]$ then 
\begin{equation*}
(M_{j}\dots M_{k-1})(M_{k}\dots M_{\ell -1})=M_{j}\dots M_{\ell -1}.
\end{equation*}
Since each of the quantities on the left are less than one, so is their
product as required. On the other hand, if $\ell \in \lbrack j,k]$, then 
\begin{align*}
(M_{j}\dots M_{k-1})(M_{k}\dots M_{\ell -1})& =M_{j}\dots M_{j-1}M_{j}\dots
M_{\ell -1} \\
& =R_{0}M_{j}\dots M_{\ell -1}<1
\end{align*}
Since $R_{0}>1$, 
\begin{equation*}
M_{j}\dots M_{\ell -1}<1,
\end{equation*}
as required.

To see that $\succ $ is asymmetric, note that $R_{0}$ can be written as 
\begin{equation*}
R_{0}=M_{0}\dots M_{n-1}=\left( M_{j}\dots M_{k-1}\right) (M_{k}\dots
M_{j-1})
\end{equation*}
Thus $R_{0}>1$ and $j\succ k$ imply $1<M_{k}\dots M_{j-1}$ contradicting $%
k\succ j.$

Finally, to see that $\succ $ is anti-reflexive, note that by definition $%
M_{j j}=1$. 
\end{proof}

\begin{definition}
We will say that $S_{j}$ is \emph{{\ unstarvable} }if $j$ is $\succ $
-maximal. Otherwise $S_{j}$ is\emph{\ \emph{starvable}} . We will refer to
the set of indices of unstarvable respectively starvable stages as $Str(\pi
) $ respectively $Unstr(\pi )$.
\end{definition}

\begin{remark}
\label{NoComparisonNoStarvation} Note that if $\pi $ is such that no two
stages are comparable, then $\succ $ is empty and every stage is $\succ $
-maximal, so $\textup{Unstr}(\pi )=[0,n)$. In particular, $\textup{Unstr}
(\pi )\neq \emptyset $ $\forall $ $\pi .$
\end{remark}

\section{Stability of infected equilibria}

\label{sec:stability}

In this section, we consider biologically meaningful infected fixed points.
In particular, we can assume $R_{0}>1$. We will show that for generic $\pi $
there is exactly one asymptotically stable biologically meaningful fixed
point. We give two characterizations of this fixed point, one in terms of
pathogen populations and one in terms of immune response.

\begin{definition}
\label{def:saturatedModerated}Let $(\S ,T^{\ast })$ be an infected fixed
point. We will say that the pathogen populations are \emph{moderated at $%
(S^{\ast },T^{\ast })$} (or simply that $(S^{\ast },T^{\ast })$ is \emph{\
moderated}) if $\S _{j}<b$ for $j\in \mathop{\textup {Unreg}}(S^{\ast
},T^{\ast })$. We will say that the immune response is \emph{saturated at }$%
(S^{\ast },T^{\ast })$ (or simply that $(S^{\ast },T^{\ast })$ is \emph{\
saturated}) if $\mathop{\textup {Reg}}(S^{\ast },T^{\ast })=\textup{Unstr}
(\pi )$.
\end{definition}

\begin{remark}
Given a fixed point $(S^{\ast },T^{\ast })$, $\S _{j}=b$ for $j\in %
\mathop{\textup {Reg}}(S^{\ast },T^{\ast })$. Thus at a moderated fixed
point, at no stage is the pathogen population greater than the regulated
population.
\end{remark}

\begin{proposition}
\label{prop:saturatedEqualsModerated}Let the set of parameters $\pi $ be
generic. Furthermore, let $(S^{\ast },T^{\ast })$ be a biologically
meaningful infected fixed point. Then the pathogen populations are moderated
at $(S^{\ast },T^{\ast })$ if and only if the immune response is saturated
at $(S^{\ast },T^{\ast }).$
\end{proposition}

\begin{proof}
We start by proving that saturation implies moderation. Suppose $(S^{\ast
},T^{\ast })$ is saturated, and suppose $\pi $ is generic. If $%
\mathop{\textup
{Unreg}}(S^{\ast },T^{\ast })=\emptyset ,$ the claim holds vacuously. So
suppose $k\in \mathop{\textup {Unreg}}(S^{\ast },T^{\ast }).$ We then have $%
k\in \textup{Str}(\pi )$. Let $j$ be $\succ $-maximal such that $j\succ k$.
If $[j+1,k)\subseteq \mathop{\textup
{Unreg}}(S^{\ast },T^{\ast })$, we are done, for then $\S _{k}=\S %
_{j}M_{jk}=bM_{jk}<b$. On the other hand if $[j+1,k)\cap 
\mathop{\textup
{Reg}}(S^{\ast },T^{\ast })\neq \emptyset $, choose $m\in \lbrack j+1,k)\cap %
\mathop{\textup {Reg}}(S^{\ast },T^{\ast })$ so that $[m+1,k)\subset %
\mathop{\textup {Unreg}}(S^{\ast },T^{\ast })$. Since $M_{jk}=M_{jm}M_{mk}<1$
, we must either have $M_{mk}<1$ forcing $\S _{k}<b$, or we have $M_{jm}<1$.
But the latter is impossible since $m\in \mathop{\textup {Reg}}(S^{\ast
},T^{\ast })$ and hence $m\in \textup{Unstr}(\pi )$. Hence $(S^{\ast
},T^{\ast })$ is moderated as required.

Now we show the converse. Suppose $(S^{\ast },T^{\ast })$ is moderated. We
then have 
\begin{align*}
\mathop{\textup {Reg}}(S^{\ast },T^{\ast })& =\{j\mid \S _{j}=b\} \\
\mathop{\textup {Unreg}}(S^{\ast },T^{\ast })& =\{j\mid \S _{j}<b\}
\end{align*}
We claim that $\mathop{\textup {Unreg}}(S^{\ast },T^{\ast })\subseteq 
\textup{Str}(\pi ).$ If $\mathop{\textup {Unreg}}(S^{\ast },T^{\ast
})=\emptyset ,$ this holds vacuously. Now assume $\mathop{\textup {Unreg}}
(S^{\ast },T^{\ast })\neq \emptyset $ and let $k\in \mathop{\textup {Unreg}}
(S^{\ast },T^{\ast })$. Choose $j\in \mathop{\textup {Reg}}(S^{\ast
},T^{\ast })$ such that $[j+1,k]\subset \mathop{\textup {Unreg}}(S^{\ast
},T^{\ast })$. Then $\S _{k}=\S _{j}M_{jk}=bM_{jk}<b$ proving $j\succ k$ as
required.

We claim that $\mathop{\textup {Reg}}(S^{\ast },T^{\ast })\subseteq \textup{%
\ Unstr}(\pi )$. Suppose to the contrary that $k\in \mathop{\textup {Reg}}
(S^{\ast },T^{\ast })\cap \textup{Str}(\pi )$. Then there is $j\in \textup{%
\ Unstr}(\pi )$ such that $j\succ k$. Since $\mathop{\textup {Unreg}}%
(S^{\ast },T^{\ast })\subset \textup{Str}(\pi )$, we must have $j\in 
\mathop{\textup
{Reg}}(S^{\ast },T^{\ast })$. But if $j\in \mathop{\textup {Reg}}(S^{\ast
},T^{\ast })$ and $j\succ k$, then, by Proposition \ref{prop:starves}, $k\in %
\mathop{\textup {Unreg}}(S^{\ast },T^{\ast })$ contradicting $k\in %
\mathop{\textup {Reg}}(S^{\ast },T^{\ast })$.

Since $\mathop{\textup {Unreg}}(S^{\ast },T^{\ast })\subseteq \textup{Str}
(\pi )$ and $\mathop{\textup {Reg}}(S^{\ast },T^{\ast })\subseteq \textup{\
Unstr}(\pi )$ and $\textup{Str}(\pi ),$ $\textup{Unstr}(\pi )$ as well as $%
\mathop{\textup {Reg}}(S^{\ast },T^{\ast }),$ $\mathop{\textup {Unreg}}
(S^{\ast },T^{\ast })$ are partitions of $[0,n)$, we have $%
\mathop{\textup
{Unreg}}(S^{\ast },T^{\ast })=\textup{Str}(\pi )$ and $%
\mathop{\textup
{Reg}}(S^{\ast },T^{\ast })=\textup{Unstr}(\pi )$ proving that $(S^{\ast
},T^{\ast })$ is saturated as required. 
\end{proof}

We now justify the use of the term ``saturated''.

\begin{theorem}
\label{thm:unsaturatedUnstable}Let the set of parameters $\pi $ be generic.
Furthermore, suppose $(S^{\ast },T^{\ast })$ is a biologically meaningful
infected fixed point and assume $(S^{\ast },T^{\ast })$ is not saturated.
Then there is a $j\in \lbrack 0,n)$ so that $T_{j}^{\ast }=0$ and for any
neighborhood $U$ of $(S^{\ast },T^{\ast })$, there is a biologically
meaningful point $x\in U$ so that $\dot{T}_{j}|_{x}>0$. In particular, $%
(S^{\ast },T^{\ast })$ is unstable.
\end{theorem}

\begin{proof}
Since $(S^{\ast },T^{\ast })$ is not saturated, by Proposition \ref%
{prop:saturatedEqualsModerated}, it is not moderated. Furthermore, $%
\mathop{\textup {Unreg}}(S^{\ast },T^{\ast })\neq \emptyset .$ Were $%
\mathop{\textup {Unreg}}(S^{\ast },T^{\ast })=\emptyset ,$ then the fixed
point would be completely regulated (${\textup{Reg}}(S^{\ast },T^{\ast
})=[0,n)$) and such a fixed point is always saturated. To see this, assume
that $(S^{\ast },T^{\ast })$ is not saturated. Since ${\textup{Unreg}}
(S^{\ast },T^{\ast })=\emptyset $, this assumption is equivalent to $%
\textup{Str}(\pi )\neq \emptyset .$ Thus, there is a $k\in \textup{Str}
(\pi )$ which is dominated by another stage $j,$ in other words, $j\succ k.$
But $j\in {\textup{Reg}}(S^{\ast },T^{\ast })$ and therefore, by
Proposition \ref{prop:starves}, $k\in {\textup{Unreg}}(S^{\ast },T^{\ast
}), $ a contradiction.

Summarizing, there is $j\in \mathop{\textup {Unreg}}(S^{\ast },T^{\ast })$
so that $\S _{j}>b$. Let $e_{T_{j}}$ be the unit vector in the $T_{j}$
direction. Recall that $\frac{\partial G_{j}}{T_{j}}(S^{\ast },T^{\ast })=(%
\S _{j}-b)$.  Thus, for any sufficiently small $\delta >0$, we have $\dot{T}%
_{j}|_{F+\delta e_{T_{j}}}>0$. This implies that arbitrarily close to $%
(S^{\ast },T^{\ast })$, there are orbits which move away from $(S^{\ast
},T^{\ast })$. 
\end{proof}

\begin{remark}
One can think of $(S^{\ast },T^{\ast })$ as a state of a \textquotedblleft
micro-ecology\textquotedblright\ in which $T_{j}$ plays the role of a
species which is capable of invading.
\end{remark}

\begin{theorem}
\label{thm:moderatedImpliesStable} Let $(S^{\ast },T^{\ast })$ be a
biologically meaningful infected fixed point. Assume that not all $T_{j}$
are equal. If $(S^{\ast },T^{\ast })$ is moderated then $(S^{\ast },T^{\ast
})$ is a local asymptotically stable equilibrium. In particular, the
eigenvalues of the Jacobian matrix $J(S^{\ast },T^{\ast })$ have strictly
negative real part.
\end{theorem}

\begin{proof}
For brevity, we will use the notations ${{\mathcal{R}}}={\ 
\mathop{\textup
{Reg}}}(S^{\ast },T^{\ast })$, ${{\mathcal{U}}}=\mathop{\textup {Unreg}}
(S^{\ast },T^{\ast })$, $\rho _{j}=r_{j}f_{j}$ and $\sigma
_{j}=a_{j}+f_{j}+T_{j}^{\ast }$ (thus, $\sigma _{j}=a_{j}+f_{j}$ for all $%
j\in {\textup{Unreg}}(S^{\ast },T^{\ast })$). The characteristic polynomial
of $J(S^{\ast },T^{\ast })$ is given by 
\begin{equation*}
P(\lambda )=\pm \prod\limits_{j=0}^{n-1}\left( (S_{j}^{\ast }-\sigma
_{j}-\lambda )(b+\lambda )-S_{j}^{\ast }(-(a_{j}+f_{j})+b)\right)
-\prod\limits_{j=1}^{n}\left( \rho _{j-1}(b+\lambda -S_{j}^{\ast })\right)
\end{equation*}
By grouping the factors containing regulated resp. unregulated stages we
obtain 
\begin{eqnarray*}
P(\lambda ) &=&\pm \prod\limits_{j\in {{\mathcal{R}}}}\left( \lambda
^{2}+\sigma _{j}\lambda +T_{j}^{\ast }b\right) \prod\limits_{j\in {\mathcal{%
U }}}\left( \lambda ^{2}+(a_{j}+f_{j}+b-S_{j}^{\ast })\lambda
+(a_{j}+f_{j})(b-S_{j}^{\ast })\right) \\
&&-\lambda ^{\left\vert {{\mathcal{R}}}\right\vert }\prod\limits_{j\in {\ {\ 
\mathcal{R}}}}\left( \rho _{j-1}\right) \prod\limits_{j\in {{\mathcal{U}}}
}\left( \rho _{j-1}(\lambda +b-S_{j}^{\ast })\right)
\end{eqnarray*}
As in Section \ref{sec:R0}, we consider the polynomials 
\begin{eqnarray}
Q(\lambda ) &:&=\pm \prod\limits_{j\in {{\mathcal{R}}}}\left( \lambda
^{2}+\sigma _{j}\lambda +T_{j}^{\ast }b\right) \prod\limits_{j\in {\mathcal{%
U }}}\left( \lambda ^{2}+(a_{j}+f_{j}+b-S_{j}^{\ast })\lambda
+(a_{j}+f_{j})(b-S_{j}^{\ast })\right)  \label{eq:EQ} \\
E(\lambda ) &:&=-\lambda ^{\left\vert {{\mathcal{R}}}\right\vert
}\prod\limits_{j\in {{\mathcal{R}}}}\left( \rho _{j-1}\right)
\prod\limits_{j\in {{\mathcal{U}}}}\left( \rho _{j-1}(\lambda +b-S_{j}^{\ast
})\right) .  \notag
\end{eqnarray}
Thus, $P=Q+E,$ and we can interpret $P$ as a perturbation of $Q$ by $E.$ By
our assumptions on the parameters, by the positivity of all $T_{j}^{\ast }$
for $j\in {\mathcal{R}}$, and by the assumption $S_{j}^{\ast }<b$ for all $%
j\in {{\mathcal{U}}}$, $Q$ is strictly Hurwitz (this is also correct in the
case ${{\mathcal{U}}} =\emptyset $). Here the question arises as to how big
the perturbation $E$ can be, so that $P$ remains strictly Hurwitz. By
Corollary 4 in \cite{Lin'sPaperOnRobustness}, the strict Hurwitz property
is preserved if $\max \left( \left\vert E(i\omega )/Q(i\omega )\right\vert
\right) <1$. \footnote{%
It is easy to verify that $\left\vert E(0)/Q(0)\right\vert =0$ and $%
\lim_{\omega \rightarrow \infty }\left\vert E(i\omega )/Q(i\omega
)\right\vert =0.$ Notice that $Q(i\omega )$ is non-zero on $(0,\infty )$.
Since $h(\omega ):=\left\vert E(i\omega )/Q(i\omega )\right\vert $ is a
non-negative continuous function on $[0,\infty )$ (even differentiable on $%
(0,\infty )$) and obviously not constant, it must take its maximum in the
open interval $(0,\infty ).$} We have 
\begin{eqnarray*}
\left\vert E(i\omega )/Q(i\omega )\right\vert &=&\frac{\left\vert (i\omega
)^{\left\vert {{\mathcal{R}}}\right\vert }\prod\limits_{j\in {{\mathcal{R}}}
}\left( \rho _{j-1}\right) \prod\limits_{j\in {{\mathcal{U}}}}\left( \rho
_{j-1}(i\omega +b-S_{j}^{\ast })\right) \right\vert }{\left\vert
\prod\limits_{j\in {{\mathcal{R}}}}\left( (i\omega )^{2}+i\sigma _{j}\omega
+T_{j}^{\ast }b\right) \prod\limits_{j\in {{\mathcal{U}}}}\left( (i\omega
)^{2}+i(a_{j}+f_{j}+b-S_{j}^{\ast })\omega +(a_{j}+f_{j})(b-S_{j}^{\ast
})\right) \right\vert } \\
&=&\frac{\omega ^{\left\vert {{\mathcal{R}}}\right\vert
}\prod\limits_{j=0}^{n-1}\left( \rho _{j}\right) \prod\limits_{j\in {\ {\ 
\mathcal{U}}}}\sqrt{\omega ^{2}+(b-S_{j}^{\ast })^{2}}}{\prod\limits_{j\in {%
\ \ {\mathcal{R}}}}\sqrt{\left( T_{j}^{\ast }b-\omega ^{2}\right)
^{2}+\sigma _{j}^{2}\omega ^{2}}\prod\limits_{j\in {{\mathcal{U}}}}\sqrt{%
\left( (a_{j}+f_{j})(b-S_{j}^{\ast })-\omega ^{2}\right)
^{2}+(a_{j}+f_{j}+b-S_{j}^{\ast })^{2}\omega ^{2}}} \\
&=&\frac{\omega ^{\left\vert {{\mathcal{R}}}\right\vert
}\prod\limits_{j=0}^{n-1}\left( \rho _{j}\right) \prod\limits_{j\in {\ {\ 
\mathcal{U}}}}\sqrt{\omega ^{2}+(b-S_{j}^{\ast })^{2}}}{\prod\limits_{j\in {%
\ \ {\mathcal{R}}}}\omega \sqrt{\left( T_{j}^{\ast }b-\omega ^{2}\right)
^{2}/\omega ^{2}+\sigma _{j}^{2}}\prod\limits_{j\in {{\mathcal{U}}}}\sqrt{
((b-S_{j}^{\ast })^{2}+\omega ^{2})((a_{j}+f_{j})^{2}+\omega ^{2})}} \\
&=&\frac{\prod\limits_{j=0}^{n-1}\left( \rho _{j}\right) }{
\prod\limits_{j\in {{\mathcal{R}}}}\sqrt{\left( T_{j}^{\ast }b-\omega
^{2}\right) ^{2}/\omega ^{2}+\sigma _{j}^{2}}\prod\limits_{j\in {{\mathcal{U}
}}}\sqrt{(a_{j}+f_{j})^{2}+\omega ^{2}}}
\end{eqnarray*}
Each of the functions $h_{j}(\omega ):=\left( T_{j}^{\ast }b-\omega
^{2}\right) ^{2}/\omega ^{2},$ $j\in {{\mathcal{R}}}$ has a global minimum
on $(0,\infty )$ at $\omega =\sqrt{T_{j}^{\ast }b}>0$ whereas the function $%
g(\omega ):=(a_{j}+f_{j})^{2}+\omega ^{2}$ has a global minimum on $%
[0,\infty )$ at $\omega =0.$ Therefore, we can conclude that for $%
\max_{\omega \geq 0}\left\vert E(i\omega )/Q(i\omega )\right\vert $, if
existent, it must hold 
\begin{equation*}
\max_{\omega \geq 0}\left\vert E(i\omega )/Q(i\omega )\right\vert <\frac{
\prod\limits_{j=0}^{n-1}\left( \rho _{j}\right) }{\prod\limits_{j\in {\ {\ 
\mathcal{R}}}}\left( \sigma _{j}\right) \prod\limits_{j\in {{\mathcal{U}}}
}(a_{j}+f_{j})}
\end{equation*}
This is also correct in the case ${{\mathcal{U}}}=\emptyset ,$ in which $%
\left\vert E(i\omega )/Q(i\omega )\right\vert
=\prod\limits_{j=0}^{n-1}\left( \rho _{j}\right) /\prod\limits_{j\in {\ {\ 
\mathcal{R}}}}\sqrt{\left( T_{j}^{\ast }b-\omega ^{2}\right) ^{2}/\omega
^{2}+\sigma _{j}^{2}}$ and $\prod\nolimits_{j\in {{\mathcal{U}}}
}(a_{j}+f_{j})=1.$ The strict inequality arises from the fact that not all $%
T_{j}^{\ast }$ are equal (and consequently not all $\sqrt{T_{j}^{\ast }b}$
are equal) and thus, the functions $h_{j}(\omega )$ do not achieve their
minima simultaneously. By 5. in Proposition \ref{prop:R0}, it follows $%
\max_{\omega \geq 0}\left\vert E(i\omega )/Q(i\omega )\right\vert <1$ and
the claim follows by Corollary 4 in \cite{Lin'sPaperOnRobustness}.
\end{proof}

\begin{corollary}
Let the set of parameters $\pi $ be generic. Furthermore, let $(S^{\ast
},T^{\ast })$ be a biologically meaningful infected fixed point. Assume that
not all $T_{j}$ are equal. Then $(S^{\ast },T^{\ast })$ is locally
asymptotically stable if and only if $(S^{\ast },T^{\ast })$ is moderated.
\end{corollary}

\begin{proof}
The claim follows immediately from Proposition \ref%
{prop:saturatedEqualsModerated} and Theorems \ref{thm:unsaturatedUnstable}
and \ref{thm:moderatedImpliesStable}.
\end{proof}

The previous Corollary can be restated as follows.

\begin{theorem}
\label{thm:saturatedImpliesStable}Let the set of parameters $\pi $ be
generic. Furthermore, let $(S^{\ast },T^{\ast })$ be a biologically
meaningful infected fixed point. Then $(S^{\ast },T^{\ast })$ is locally
asymptotically stable if and only if $(S^{\ast },T^{\ast })$ is saturated.
\end{theorem}

\begin{proof}
Under the assumption that $\pi $ is generic, $(S^{\ast },T^{\ast })$ is
saturated if and only if it is moderated. Note that for a generic parameter
set $\pi $, the $T_{j}^{\ast }$ are not all equal as this occurs only on a
subset of parameters of co-dimension $n-1$. 
\end{proof}

\begin{remark}
\label{rem:ThirdGenericityAssumption} We have introduced a third requirement
for $\pi $ to be generic, namely that $T_{j}^{\ast }$ are not all equal. As
we will see below, this is equivalent to the requirement that $%
r_{j-1}f_{j-1}-(a_{j}+f_{j})$ is not independent of $j$.
\end{remark}

\begin{definition}
Let $(S^{\ast },T^{\ast })$ be a fixed point. If $\mathop{\textup {Reg}}
(S^{\ast },T^{\ast })=[0,n)$, we will say that $(S^{\ast },T^{\ast })$ is 
\emph{completely regulated}.
\end{definition}

\begin{remark}
If $(S^{\ast },T^{\ast })$ is completely regulated, it follows that $%
(S^{\ast },T^{\ast })$ is infected. It is also vacuously moderated. If $\pi $
is generic and $(S^{\ast },T^{\ast })$ biologically meaningful, it is also
saturated (see the Proof of Theorem \ref{thm:unsaturatedUnstable}). On the
other hand, if $\pi $ is generic and such that no two stages are comparable,
then the completely regulated (not necessarily biologically meaningful)
fixed point is saturated (cf. Remark \ref{NoComparisonNoStarvation}).
\end{remark}

\begin{corollary}
\label{Cor.CompletelyRegulatedFPisAlwaysStable} Let $(S^{\ast },T^{\ast })$
be the completely regulated biologically meaningful fixed point. Suppose
that the $T_{j}^{\ast }$ are not all equal. Then $(S^{\ast },T^{\ast })$ is
locally asymptotically stable and the eigenvalues of the Jacobian matrix $%
J(S^{\ast },T^{\ast })$ have strictly negative real part.
\end{corollary}

\begin{proof}
Since $(S^{\ast },T^{\ast })$ is completely regulated, $(S^{\ast },T^{\ast
}) $ is moderated. 
\end{proof}

\begin{remark}
Note that if we drop the assumption that not all the $T_{j}^{\ast }$ are
equal, we are forced into the case $T_{j}^{\ast }=\tau $ for all $j\in
\lbrack 0,n)$ for some $\tau \in \mathbb{R}^{+}$ and consequently $%
\mathop{\textup {Unreg}}(S^{\ast },T^{\ast })=\emptyset $. In this case, 
\begin{equation*}
\max_{\omega \geq 0}\left\vert E(i\omega )/Q(i\omega )\right\vert
=\max_{\omega \geq 0}\left( \frac{\prod\limits_{j=0}^{n-1}\left(
p_{j}\right) }{\prod\limits_{j=0}^{n-1}\sqrt{\left( \tau b-\omega
^{2}\right) ^{2}/\omega ^{2}+s_{j}^{2}}}\right) =\frac{\prod
\limits_{j=0}^{n-1}\left( p_{j}\right) }{\prod
\limits_{j=0}^{n-1}(a_{j}+f_{j}+T_{j}^{\ast })}
\end{equation*}
and by Proposition \ref{prop:R0}, $\max_{\omega \geq 0}\left\vert E(i\omega
)/Q(i\omega )\right\vert =1.$ Corollary 4 in \cite{Lin'sPaperOnRobustness}
and the continuity of the eigenvalues on the entries of a matrix let us
conclude that $\textup{Re}(\lambda _{k})\leq 0$ for all roots $\lambda _{k}$
of the characteristic polynomial $P.$ In this very non-generic case,
linearization of the right hand side of (\ref{Eq.BasicModel}) does not yield
a stability statement for $(S^{\ast },T^{\ast })$. In the next subsection we
will study this case in more detail.
\end{remark}

\section{An embedded Lotka-Volterra system}

\label{sec:Lotka-Volterra}

For certain highly non-generic sets of parameters, this system contains an
embedded Lotka-Volterra predator-prey system. Recall that this system is
given by 
\begin{align*}
\dot{x}& =(\alpha -y)x \\
\dot{y}& =(x-\beta )y
\end{align*}
where $\alpha ,\beta \in \mathbb{R}^{+}$ are positive parameters and $x$ and 
$y$ are non-negative. (See, e.g., \cite{hirschSmale}.) We will use the
notation $\Delta _{S}\times \Delta _{T} $ to denote points $(S_{0},\dots
,S_{n-1},T_{0},\dots ,T_{n-1})$ where the values $S_{j}$ and $T_{j}$ are
independent of $j$, that is, points satisfying $S_{0}=S_{1}=\dots =S_{n-1}$
and $T_{0}=T_{1}=\dots =T_{n-1}.$

\begin{theorem}
\label{thm:Lotka-Volterra}Suppose $(S^{\ast },T^{\ast })$ is the completely
regulated infected fixed point of (\ref{Eq.BasicModel}). Then the following
are equivalent:

\begin{enumerate}
\item \label{a}$T_{j}^{\ast }$ is independent of $j$.

\item \label{b}$(S^{\ast },T^{\ast })\in \Delta _{S}\times \Delta _{T}$.

\item \label{c}The expression $r_{j-1}f_{j-1}-(a_{j}+f_{j})$ is positive and
independent of $j$.

\item \label{d}$\Delta _{S}\times \Delta _{T}$ is invariant and the dynamics
of the system restricted to $\Delta _{S}\times \Delta _{T}$ is the
Lotka-Volterra system.

\item \label{e}There is an embedded Lotka-Volterra system.

\item \label{f}The Jacobian of the system has a pair of pure imaginary
eigenvalues.
\end{enumerate}
\end{theorem}

\begin{proof}
~

\textit{\ref{a}.} $\implies$ \textit{\ref{b}.}: We have $T^*_j\ne 0$ for $%
0\le j < n$, so $\S _j = b$ for $0\le j < n$.

\textit{\ref{b}.} $\implies $ \textit{\ref{c}.}: Recall that 
\begin{equation*}
T_{j}^{\ast }=\frac{r_{j-1}f_{j-1}}{b}\S _{j-1}-(a_{j}+f_{j}).
\end{equation*}
Since $\S _{j}=b$ for $0\leq j<n$\ and $T_{j}^{\ast }$ is independent of $j$
, the result follows.

\textit{\ref{c}.} $\implies $ \textit{\ref{d}.}: We take $S_{j}=x$, $T_{j}=y$
on $\Delta _{S}\times \Delta _{T}$, $r_{j-1}f_{j-1}-(a_{j}+f_{j})=\alpha $
and $b=\beta $. We then have 
\begin{align*}
\dot{S}_{j}& =r_{j-1}f_{j-1}S_{j-1}-(a_{j}+f_{j}+T_{j})S_{j} \\
& =(\alpha -y)x \\
\dot{T}_{j}& =(S_{j}-b)T_{j} \\
& =(x-\beta )y
\end{align*}
Since the right-hand sides are independent of $j$, we may replace the
left-hand sides with the derivatives of $x$ and $y$ giving the
Lotka-Volterra equations.

\textit{\ref{d}.} immediately implies \textit{\ref{e}.}

\textit{\ref{f}.} now follows from the well-known dynamics of the
Lotka-Volterra equations.

It remains to prove that \textit{\ref{f}.} implies \textit{\ref{a}.} This
follows from Corollary~\ref{Cor.CompletelyRegulatedFPisAlwaysStable}. 
\end{proof}

\begin{remark}
The Lotka-Volterra dynamics take place on an embedded 2-plane. Corollary~\ref%
{Cor.CompletelyRegulatedFPisAlwaysStable} says that the eigenvalues of the
Jacobian have non-positive real parts. This implies that there is no
exponential departure from the fixed point orthogonal to the Lotka-Volterra
plane. We conjecture that the pure imaginary eigenvalues of the
Lotka-Volterra dynamics are the only ones with non-negative real part. If
this conjecture is correct, the Lotka-Volterra plane is an attractor in a
neighborhood of the fixed point. This conjecture is consistent with the
results of numerical investigations.
\end{remark}

\section{Self-establishing stages}

\label{sec:selfEstablishingStages}

In this section we show how to modify Definition~\ref{def:starves} to
generalize Theorems \ref{thm:unsaturatedUnstable}, \ref%
{thm:moderatedImpliesStable} and \ref{thm:saturatedImpliesStable} in the
case where there exists a $j\in \lbrack 0,n)$ such that $a_{j}+f_{j}<0$.

In previous sections, we have assumed that $a_{j}+f_{j}>0$ for all $j\in
\lbrack 0,n)$. The assumption that $a_{j}>0$ says that stage $j$ decays
rather than proliferates. If $a_{j}<0$ and $a_{j}+f_{j}>0$, then stage $j$
proliferates, but does so more slowly than it is lost to stage $j+1$.
Accordingly, no stage is able to establish itself independently of the other
stages. We now lift this assumption and allow stages $j$ such that $%
a_{j}+f_{j}<0$. Such a stage proliferates faster than it differentiates and
is thus able to establish itself independently of the remaining stages. We
call such a stage \emph{self-establishing}.

Since $f_j$ and $r_j$ are positive for all $j$, a self-establishing stage
also establishes infection at all other stages. The converse is also true.
Suppose a naive is infected with a small amount of $S_j>0$ where $j\notin {%
\mathop{\textup {SE}}}(\pi)$. Suppose also that $k\in{\mathop{\textup {SE}}}%
(\pi)$ with $[j+1,k) \cap {\mathop{\textup {SE}}}(\pi) = \emptyset$. Given
that $S_j>0$ for $t=0$, (\ref{Eq.BasicModel}) ensures that $S_k>0$ for any
sufficiently small $t$. $S_k$ then proceeds to proliferate.

As our fourth and last genericity assumption, we assume $a_{j}+f_{j}\neq 0$ $%
\forall$ $j\in \lbrack 0,n).$ We take ${\mathop{\textup {SE}}}(\pi )=\{j\mid
a_{j}+f_{j}<0\}$. In this section, we take it as a standing assumption that $%
{\mathop{\textup {SE}}}(\pi )\neq \emptyset $.

\begin{proposition}
\label{prop: SelfEstablishmentImpliesRegulation}Let $\pi $ be a set of
parameters. Furthermore, let $(S^{\ast },T^{\ast })$ be a biologically
meaningful infected fixed point. If $j\in {\mathop{\textup {SE}}}(\pi )$,
then $j\in \mathop{\textup {Reg}}(S^{\ast },T^{\ast })$.
\end{proposition}

\begin{proof}
If, on the contrary, $T_{j}^{\ast }=0$, it follows that $\dot{S_{j}}>0$,
contradicting the assumption that $(\S ,T^*)$ is a fixed point. 
\end{proof}

It should be clear that any self-establishing stage cannot be starved of
sufficient population to support regulation.

\begin{definition}
\label{def:generalStarves} We say that $j$ \emph{starves} $k$ and write $%
j\succ k$ if $M_{j k}<1$ and $[j+1,k] \cap {\mathop{\textup {SE}}}(\pi) =
\emptyset$.
\end{definition}

Note that this reduces to the previous definition when ${\ 
\mathop{\textup
{SE}}}(\pi)$ is empty.

\begin{proposition}
Let $\pi $ be such that $R_{0}>1.$ Then the relation $\succ $ is a strict
partial order.
\end{proposition}

\begin{proof}
Let ${\mathop{\textup {SE}}}(\pi)$ consist of $j_1 < j_2 <\dots < j_m$.
These decompose $[0,n)$ into $U_1,\dots,U_{m+1}$ where $U_\ell =
[j_\ell,j_{\ell+1})$ for $\ell=1,\dots, m$ and $U_{m+1} = [j_m,j_1)$. We
order each of these segments according to the reverse of the order induced
by the cyclic order of $[0,n)$. That is, for $\ell=1,\dots,m$, $U_\ell$ is
ordered $j_\ell {\overset{\ast}{>}} j_\ell +1 {\overset{\ast}{>}} \dots {\ 
\overset{\ast}{>}} j_{\ell+1}-1$, while $U_{m+1}$ is ordered $j_m {\overset{
\ast}{>}} \dots {\overset{\ast}{>}} n-1 {\overset{\ast}{>}} 0 {\overset{\ast}%
{>}} \dots {\overset{\ast}{>}} j_1-1$. Notice that if $j\succ k$ then $j$
and $k$ lie in the same $U_\ell$ and $j {\overset{\ast}{>}} k$. It is easy
to see that this makes $\succ$ a strict partial order. 
\end{proof}

We define $\textup{Str}(\pi )$, $\textup{Unstr}(\pi )$, saturated and
moderated as before. The statement of Proposition \ref{prop:starves} as well
as the equivalence of moderation and saturation (Proposition \ref%
{prop:saturatedEqualsModerated}) and their necessity for stability (Theorem %
\ref{thm:unsaturatedUnstable}) can be proved analogously, where the result
of Proposition \ref{prop: SelfEstablishmentImpliesRegulation} plays an
important role.

In order to prove sufficiency in the presence of self-establishing stages,
we need the following lemma.

\begin{lemma}
Let the set of parameters $\pi $ be generic. Furthermore, let $(S^{\ast
},T^{\ast })$ be a biologically meaningful infected fixed point. Then

\begin{enumerate}
\item $\S _j>0$ for $j\in [0,n)$.

\item $\sigma_j = a_j + f_j + T^*_j > 0$ for $j\in[0,n)$.
\end{enumerate}
\end{lemma}

\begin{proof}
Since $(\S ,T^*)$ is infected and biologically meaningful, there is $j$ so
that $\S _j > 0$. Hence, by (\ref{Eq.FixedPointEquations1}) $r_jf_j \S _j > 0
$, which in turn requires $-(a_{j+1}+f_{j+1}+T^*_{j+1})\S _{j+1} < 0$. This
requires $\S _{j+1}\ne 0$, and since $(\S ,T^*)$ is biologically meaningful, 
$\S _{j+1}> 0$. This, in turn, requires $\sigma_{j+1} =
a_{j+1}+f_{j+1}+T^*_{j+1} > 0$. Continuing in this way extends this to all $%
j\in[0,n)$. 
\end{proof}

The proof of Theorem~\ref{thm:moderatedImpliesStable} depends on equation~ (%
\ref{eq:EQ}), and we now revisit this in the case of a saturated fixed point
with self-establishing stages. To see that $Q(\lambda)$ is Hurwitz, we must
check that its coefficients are positive. For the product taken over
regulated stages, we note that $\sigma_j$ and $T^*_j b$ are both positive.
We turn to the product taken over unregulated stages. As we have seen, an
unregulated stage cannot be self-establishing. In particular, at any
unregulated stage, $a_j+f_j$ is positive and we proceed as before. Theorem~ %
\ref{thm:saturatedImpliesStable} and Corollary~\ref%
{Cor.CompletelyRegulatedFPisAlwaysStable} now follow as before.

\section{Discussion}

\label{sec:discussion}

In this paper, we have considered an ODE model (\ref{Eq.BasicModel}) of the
interactions between a host and a pathogen which uses a cycle of
antigenically distinct stages to establish and maintain infection. We were
able to give a simple expression for $R_{0}$ which determines whether the
pathogen is able to establish infection. If $R_{0}<1$, exposure to the
pathogen, even in large quantities, will not succeed in infecting the host
since the pathogen dies faster than it can replace itself. Indeed, the host
will clear the pathogen even in the absence of an immune response.

If $R_0>1$, then the pathogen is capable of infecting the host, and in this
model, an immune response is necessary to control the infection.\footnote{
In models which include the supply of uninfected tissue, the rate of supply
provides a limiting factor for infection in the absence of immune response.}
In this case, the absence of an immune response leads to unbounded expansion
of the pathogen population and presumably, the host's death.

The question then arises, when an immunocompetent host is infected with an
infectively competent pathogen (i.e., $R_{0}>1$), can they arrive at a stable
equilibrium? If so, what are the balancing populations for pathogen population
and host response? In terms of the model, this question translates to asking
which of the model's equilibria are stable and under what conditions. A key
ingredient towards answering this question was the ability to distinguish the
different stages of the pathogen in terms of the net yield with which they are
generated. We were able to formalize this distinction using a binary relation
$h\starves k$, which turns out to be a partial order if $R_0>1$ (see Section
\ref{sec:PartialOrder}).

Assuming the pathogen can establish infection the concepts of starvable and
unstarvable stages restrict the possible patterns of regulation.  The
starvation relation $j\starves k$ is based on comparison of the rate at which
stage $j$ produces stage $k$ with the rate at which stage $k$ is lost to death
and differentiation.  This determines the ability of stage $j$ to replenish
the population at stage $k$.  When this downstream yield is smaller than the
regulated population $b$, it is insufficient to support regulation.  In this
way, if $j \starves k$ and $j$ is regulated, then regulation of stage $k$
vanishes.  In this sense, the immune regulation of unstarvable stages is
sufficient to immunologically control the starvable stages. At steady state,
immune regulation is only required against those stages that are produced with
relatively higher yield and the equilibrium becomes moderated (see Section
\ref{sec:stability}). Accordingly, there is a unique biologically meaningful
infected equilibrium, namely, the one at which all unstarvable stages are
regulated and all starvable stages are unregulated.

If the parameters of the system allow for the existence of a biologically
meaningful completely regulated fixed point, the model predicts that this will
be the (only) stable equilibrium under the premise that the pathogen can
establish infection (Corollary \ref{Cor.CompletelyRegulatedFPisAlwaysStable}).
According to equations (\ref{Eq.MultipleStepTvalues}), the completely
regulated fixed point satisfies $S_{j}^{\ast }=b\text{ and }T_{j}^{\ast
}=r_{j-1}f_{j-1}-(a_{j}+f_{j}),\text{ }j=0,...,n-1.$ The inequalities
$T_{j}^{\ast }>0$ imply that the amplification factors $r_{j}$ are big enough
for every stage to be replenished at a higher rate than the sum $a_{j}+f_{j}$
of its own decay and transformation rate. Our model predicts that a pathogen
which produces every stage with such an effectivity can only be
immunologically controlled with a fine tuned immune response against every
stage the pathogen cycles through.

Having clarified the properties of equilibria, it is pertinent to keep in mind
that these are local properties of the model. We were unable to perform a
global analysis of the dynamics (see, for instance, \cite
{GlobaAnOfVirDynamics}). That is, we cannot give an account of trajectories
along which the system might approach such a steady state configuration, the
basins of attraction of attracting fixed points, nor can we eliminate the
possibility that the system exhibits limiting cycles or other attractors.  Our
naive expectation is that the unique biologically meaningful infected
equilibrium acts as a global attractor in the case where the host possesses an
immunocompetent response to every stage.

Our methods for establishing the uniqueness and existence of the unique stable
fixed point have very distinct flavors. The arguments for uniqueness have
clear biological interpretations. In particular, the starvation relation,
$j\succ k$ of Section~\ref{sec:PartialOrder} considers the maximal rate at
which stage $j$ can produce stage $k$ with the rate at which stage $k $ is
lost to stage $k+1$ and in turn compares these to the minimum population
necessary to support immune regulation. Likewise,
Theorem~\ref{thm:unsaturatedUnstable} describes the ability of an immune
response to invade and therefore alter a given equilibrium.

Our methods for establishing existence have a very different character. Here
we have turned to the eigenvalues of the system's Jacobian matrix and these
have lead us to the roots of its characteristic polynomial. These sorts of
computations can offer many challenges which have been the subject of intense
study for the last five decades within the control and systems theory
community. Many of the resulting techniques focus on the sorts of
perturbations that can be applied to a strictly Hurwitz polynomial without
losing the Hurwitz property. We predict that as biological models continue to
grow in complexity, mathematical biologists will increasingly turn to these
analytic tools. For entry into this topic we recommend the website
\cite{controlTheoryWebsite}.

Having stated the general predictions that can be derived from our model, we
would like to elaborate on the model's features, properties and predictions
within the biological framework of chronic EBV infection. It was indeed
human EBV pathology what inspired the basic structure of our model. To this
end, we provide a brief overview of what is known about EBV infection in
humans.

EBV is a highly successful pathogen which infects over 90\% of the adult human
population \cite{EBVreview}. It is horizontally transmitted via saliva
\cite{EBVtransmission} and is tropic to the epithelium of the oropharynx and
B-cells. It infects naive B-cells and causes them to become activated
B-blasts. These enter the germinal centers of Waldeyer's ring where they are
ultimately transformed into latently infected quiescent memory B-cells. These
circulate in the peripheral blood where they rarely express viral proteins and
are therefore invisible to the immune system \cite{[15]}. Upon return to
Waldeyer's ring, some of these are triggered to start producing virus
\cite{[16]}. These lytically infected B-cells ultimately burst, producing
free virus which may be transmitted horizontally or may reinfect additional
naive B-cells, thus completing the cycle \cite {tl2004}. The epithelium of
the oropharynx appears to play a role in amplifying virus for transmission
\cite{hadinotoShedding}. The host mounts T-cell responses against the blast,
germinal center and lytic stages, and an antibody response against the free
virus. The normal course of acute infectious mononucleosis involves an acute
phase during which latently infected memory B-cells can become as much as 50\%
of the memory B-cells in the peripheral circulation \cite{[17]}. Immune
response is also high during the acute phase. This resolves to a chronic phase
of persistent low-level infection \cite{[17]}, \cite{[18]}.

This suggests that a mathematical model of EBV infection might include the
following infected populations: activated B-blasts, infected germinal center
B-cells, quiescent infected memory B-cells and lytically infected B-cells.
These latter can be subdivided into stages exhibiting distinct antigenic
profiles, namely immediate early, early and late antigens. Due to its short
life span one might choose to omit free virus from such a model. In short, it
is not at all obvious what would be the ``right'' number of stages to
include in an EBV model, and this has motivated our choice to study a model
which is entirely general as to the number of stages.

We have chosen to omit the supply of uninfected tissue from our model for the
sake of symmetry. We would argue that this does not affect the validity of the
model unless a significant fraction of the uninfected tissue is lost to
infection. To see this, consider a modification of the model to take this into
account. Let us use $S_{n}$ to model uninfected tissue, let us take $S_{n-1}$
to model free virus\footnote{ In this case $T_{n-1}$ is humoral response.},
and let the rate of new infection be proportional to their product. We then
have
\begin{align*}
\dot{S}_{n-1}& =r_{n-2}f_{n-2}S_{n-2}-(a_{n-1}+\beta S_{n}+T_{n-1})S_{n-1} \\
\dot{S}_{n}& =\lambda -(a_{n}+r_{n}\beta S_{n-1})S_{n} \\
\dot{S}_{0}& =r_{n}\beta S_{n}S_{n-1}-(a_{0}+f_{0}+T_{0})S_{0}
\end{align*}
(Compare equation 3.1 of \cite{NowakMay}.) If $S_{n}$ is relatively
constant, then $S_{n}$ can be eliminated from the model and we can take
$f_{n-1}=\beta S_{n}$ and $r_{n-1}f_{n-1}=r_{n}\beta S_{n}$.

While acute phase EBV infection shows infected cells occupying a large
portion of peripheral memory, it seems unlikely that it consumes a large
portion of the uninfected naive B-cells supplied to Waldeyer's ring\cite 
{[15]}, \cite{souza05}. Accordingly, it seems unlikely that this will be an
obstruction to application of this model.

Our primary purpose in developing this model is to cast light on the
long-term persistence of EBV infection. Here we can surely treat the supply
of uninfected tissue as nearly constant.

Biologists will be interested in solutions to (\ref{Eq.BiologicalModel})
rather than (\ref{Eq.BasicModel}). The follow-on factors are unchanged. The
analogs of equations (~\ref{Eq.MultipleStepTvalues}) and (\ref
{Eq.MultipleStepSvalues}) are 
\begin{align*}
\S _j&=\frac{b}{c_j} \text{ and } T^*_j = \frac{1}{p_j}
\left(\frac{c_j}{c_{h_{j-1}}} r_{j-1}f_{j-1} M_{h_j j-1} - (a_j+f_j)\right)
\text{ } \forall \text{ } j\in\mathop{\textup {Reg}}(\S ,T^*) \\ 
\S _j &=
\frac{b}{c_{h_j}} M_{h_j j} \text{ and } T^*_j=0 \text{ } \forall \text{ }
j\in\mathop{\textup {Unreg}}(\S ,T^*)
\end{align*}
The starvation relation is now $j\succ k$ if $M_{jk} < \frac{c_j}{c_k}$.

People show different patterns of infection and immune response to EBV.  Under
10\% of the adult population is EBV negative. Given that over 90\% of the
adult population is EBV positive \cite{EBVreview}, \cite{BasicEBV_Virology},
it is clear that most of the EBV negative population has had repeated exposure
to EBV. Since these people show no EBV-specific antibodies \cite{henle}, it
seems clear that they are negative due to the virus's inability to establish
infection rather than the immune system's ability to clear it. This shows that
$R_{0}$ is a host-pathogen property, not simply a property of the pathogen.

Chronically infected people show very stable levels of infected memory
B-cells in the peripheral blood \cite{[18]}. This is an indication that we
are dealing with a stable fixed point of the host-pathogen system. However
the level of infected memory B-cells varies from person to person \cite
{[18]} as does the pattern of immune response. To the best of our
knowledge, there is always a humoral response to the free virus. Beyond
that, we have seen the individuals who exhibit T-cell responses to both the
blast and germinal center stages and to the lytic stages and individuals who
only exhibit detectable T-cell response to the lytic stages \cite{TomG}.
The latter patterns could arise from a gap in the available T-cell
repertory. We think it is at least as likely that the starvation
relationship varies due to population variation in $\pi$. Empirical
determination of these parameters is work in progress.

It is instructive to consider the variation levels of latently infected
memory B-cells. To the best of our knowledge, this compartment is never
regulated \cite{[17]}. Due to the very low levels of viral protein
expression, this stage has a very low overall antigenicity as compared to
any other stage. Taking $S_{\text{latent}}$ to be latently infected memory
B-cells, we see that a low value for $c_{\text{latent}}$ compared to other
stages implies starvability since we then have $M_{j {\text{latent}}} < 
\frac{c_j}{c_{\text{latent}}}$.

The expression $\S _j = \frac{b}{c_j}$, shows that the size of regulated
populations depends on factors governing immune response, not on the flow of
cells into that population. In particular, we can expect the size of the
lytically infected B-cell population to show little or no correlation with
the size of the latently infected memory B-cell population, while it might
well be a correlate with immune factors such as HLA types.

The question arises whether to expect a correlation between
$T_{\text{lytic}}^{\ast }$ and $\S _{\text{latent}}$. Here the picture is not
clear as the model's predictions depend on the reasons for the variation in
$\S _{\text{ latent}}$. If high values of $\S _{\text{latent}}$ are due to a
high rate flow into this compartment from the infected germinal center
population, then we would expect a positive correlation between
$T_{\text{lytic}}^{\ast } $ and $\S _{\text{latent}}$. If, on the other hand,
high values of $\S _{ \text{latent}}$ are due to low rate of flow from the
latently infected memory B-cell compartment into the lytically infected
population, then we would expect a negative correlation between
$T_{\text{lytic}}^{\ast }$ and $ \S _{\text{latent}}$.

On a different vein, from the point of view of the treatment of chronic
infections, the characteristics of the immune response revealed by our
model, in particular, the possibility of unregulated stages, give rise to
speculation on the possible targets of such a treatment. Which stages would
be most suitable drug targets, regulated or unregulated ones?

Finally, we would like say a few words about some mathematical questions
raised by this model. Our stability results are all local. We would
conjecture that the unique asymptotically stable biologically meaningful
infected fixed point is a global attractor for the positive orthant. This
would follow from the existence of an appropriate Lyapunov function.
However, we have not yet managed to find one.

We also conjecture that in the Lotka-Volterra case, the 2-plane that carries
the Lotka-Volterra dynamics is an attractor, at least in a neighborhood of
the fixed point. Our numerical investigations to date support this
conjecture.

There are a number of generalizations of this model which deserve
investigation. These include the supply of uninfected tissue, and immune
cross-reactivity in which an immune response might target multiple stages.

One interesting generalization concerns the situation in which each stage
may be regulated by an immune response to more than one epitope and in which
each of these epitopes occurs in multiple variants. This question is
addressed for single-stage pathogens in \cite{nowakMultipleEpitopes}.
Indeed Nowak suggested \cite{NowakEmail} that we examine this question for
multi-stage pathogens and we have preliminary results in the case where only
one stage exhibits this phenotype.

\newpage

\section*{Appendix}

\label{sec:appendix}

\subsection*{Derivation of an expression for the characteristic polynomial
of $J(S,J)$}

\begin{theorem}
\label{thm:charicteristicPolynomial} The characteristic polynomial of the
Jacobian $J(S,T)$ of system~(\ref{Eq.BasicModel}) is given by 
\begin{equation*}
P(\lambda )=(-1)^{n}\prod\limits_{j=0}^{n-1}\left( (S_{j}-\sigma
_{j}-\lambda )(b+\lambda )-S_{j}(-(a_{j}+f_{j})+b)\right)
-\prod\limits_{j=0}^{n-1}\left( \rho _{j}(b+\lambda -S_{j+1})\right) .
\end{equation*}
\end{theorem}

\begin{proof}
The partial derivatives of the system are given by 
\begin{align*}
\frac{\partial F_{k}}{\partial S_{j}}& = 
\begin{cases}
r_{k-1}f_{k-1} & \text{if $j=k-1$} \\ 
-a_{k}-f_{k}-T_{k} & \text{if $j=k$} \\ 
0 & \text{otherwise}%
\end{cases}
\\
\frac{\partial F_{k}}{\partial T_{j}}& = 
\begin{cases}
-S_{k}~~~~~~~~~~~~~~~ & \text{if $j=k$} \\ 
0 & \text{otherwise}%
\end{cases}
\\
\frac{\partial G_{k}}{\partial S_{j}}& = 
\begin{cases}
T_{k}~~~~~~~~~~~~~~~~~~ & \text{if $j=k$} \\ 
0 & \text{otherwise}%
\end{cases}
\\
\frac{\partial G_{k}}{\partial T_{j}}& = 
\begin{cases}
S_{k}-b~~~~~~~~~~~~~ & \text{if $j=k$} \\ 
0 & \text{otherwise}%
\end{cases}%
\end{align*}
Thus, we have 
\begin{equation*}
J(S,T)=\left( 
\begin{matrix}
A & B \\ 
C & D%
\end{matrix}
\right)
\end{equation*}
where $A,B,C$ and $D$ are the following $n\times n$ matrices 
\begin{align*}
A& =\left( 
\begin{matrix}
-\sigma _{0} & 0 & \cdots & 0 & \rho _{n-1} \\ 
\rho _{0} & -\sigma _{1} & \ddots &  & 0 \\ 
0 & \ddots & \ddots & \ddots & \vdots \\ 
\vdots &  & \ddots & \ddots & 0 \\ 
0 & \cdots & 0 & \rho _{n-2}~ & ~-\sigma _{n-1}%
\end{matrix}
\right) \\
B& =\left( 
\begin{matrix}
-S_{0} & 0 & \cdots & 0 \\ 
0 & \ddots &  & 0 \\ 
\vdots &  & \ddots & \vdots \\ 
0 & \cdots & 0 & -S_{n-1}%
\end{matrix}
\right) \\
C& =\left( 
\begin{matrix}
T_{0} & 0 & \cdots & 0 \\ 
0 & \ddots &  & 0 \\ 
\vdots &  & \ddots & \vdots \\ 
0 & \cdots & 0 & T_{n-1}%
\end{matrix}
\right) \\
D& =\left( 
\begin{matrix}
S_{0}-b & 0 & \cdots & 0 \\ 
0 & \ddots &  & 0 \\ 
\vdots &  & \ddots & \vdots \\ 
0 & \cdots & 0 & S_{n-1}-b%
\end{matrix}
\right)
\end{align*}
using the notation $\rho _{j}=r_{j}f_{j}$ and $\sigma
_{j}=a_{j}+f_{j}+T_{j}. $ It is well known that a similarity transform does
not modify the eigenvalues of a matrix. The similarity transform we use
consists of the following operations: The first $n$ columns of $J(S,T)$ are
transformed to 
\begin{equation*}
\begin{pmatrix}
A^{\prime } \\ 
C^{\prime }%
\end{pmatrix}
= 
\begin{pmatrix}
A \\ 
C%
\end{pmatrix}
- 
\begin{pmatrix}
B \\ 
D%
\end{pmatrix}%
\end{equation*}
The resulting matrix 
\begin{equation*}
\left( 
\begin{matrix}
A^{\prime } & B \\ 
C^{\prime } & D%
\end{matrix}
\right)
\end{equation*}
undergoes the following row transformation 
\begin{equation*}
\begin{pmatrix}
C^{\prime \prime } & D^{\prime }%
\end{pmatrix}
= 
\begin{pmatrix}
C^{\prime } & D%
\end{pmatrix}
+ 
\begin{pmatrix}
A^{\prime } & B%
\end{pmatrix}%
\end{equation*}
The resulting matrix 
\begin{equation*}
\left( 
\begin{matrix}
A^{\prime } & B \\ 
C^{\prime \prime } & D^{\prime }%
\end{matrix}
\right)
\end{equation*}
then undergoes the following column permutations 
\begin{equation*}
\left( 
\begin{matrix}
B & A^{\prime } \\ 
D^{\prime } & C^{\prime \prime }%
\end{matrix}
\right)
\end{equation*}
followed by row permutations 
\begin{equation*}
\left( 
\begin{matrix}
D^{\prime } & C^{\prime \prime } \\ 
B & A^{\prime }%
\end{matrix}
\right)
\end{equation*}
The similar matrix obtained satisfies 
\begin{align*}
A^{\prime }& =\left( 
\begin{matrix}
S_{0}-\sigma _{0} & 0 & \cdots & 0 & \rho _{n-1} \\ 
\rho _{0} & S_{1}-\sigma _{1} & \ddots &  & 0 \\ 
0 & \ddots & \ddots & \ddots & \vdots \\ 
\vdots &  & \ddots & \ddots & 0 \\ 
0 & \cdots & 0 & \rho _{n-2}~ & ~S_{n-1}-\sigma _{n-1}%
\end{matrix}
\right) \\
C^{\prime \prime }& =\left( 
\begin{matrix}
T_{0}-\sigma _{0}+b & 0 & \cdots & 0 & \rho _{n-1} \\ 
\rho _{0} & S_{1}-\sigma _{1}+b & \ddots &  & 0 \\ 
0 & \ddots & \ddots & \ddots & \vdots \\ 
\vdots &  & \ddots & \ddots & 0 \\ 
0 & \cdots & 0 & \rho _{n-2}~ & ~T_{n-1}-\sigma _{n-1}+b%
\end{matrix}
\right) \\
D^{\prime }& =\left( 
\begin{matrix}
-b & 0 & \cdots & 0 \\ 
0 & \ddots &  & 0 \\ 
\vdots &  & \ddots & \vdots \\ 
0 & \cdots & 0 & -b%
\end{matrix}
\right)
\end{align*}
Now we look at the eigenvalue equations of the transformed matrix 
\begin{equation*}
\left( 
\begin{matrix}
\widetilde{D^{\prime }} & C^{\prime \prime } \\ 
B & \widetilde{A^{\prime }}%
\end{matrix}
\right) :=\left( 
\begin{matrix}
D^{\prime } & C^{\prime \prime } \\ 
B & A^{\prime }%
\end{matrix}
\right) -\lambda I
\end{equation*}
where $I$ is the $2n\times 2n$ unity matrix. If we assume $\lambda \neq -b,$
the matrix 
\begin{equation*}
\widetilde{D^{\prime }}=\left( 
\begin{matrix}
-(b+\lambda ) & 0 & \cdots & 0 \\ 
0 & \ddots &  & 0 \\ 
\vdots &  & \ddots & \vdots \\ 
0 & \cdots & 0 & -(b+\lambda )%
\end{matrix}
\right)
\end{equation*}
is invertible and we can use a well known formula for the determinant of
block-partitioned matrices 
\begin{align*}
Q(\lambda ) &:= \det \left( 
\begin{matrix}
\widetilde{D^{\prime }} & C^{\prime \prime } \\ 
B & \widetilde{A^{\prime }}%
\end{matrix}
\right) =\det (\widetilde{D^{\prime }})\det (\widetilde{A^{\prime }}-B 
\widetilde{D^{\prime }}^{-1}C^{\prime \prime }) \\
&=(-1)^{n}(b+\lambda )^{n} \\
&\ \ \ \ \det \left( 
\begin{matrix}
S_{0}-\sigma _{0}-\lambda -S_{0}\frac{(T_{0}-\sigma _{0}+b)}{b+\lambda } & 0
& \cdots & 0 & \rho _{n-1}-\frac{S_{0}\rho _{n-1}}{b+\lambda } \\ 
\rho _{0}-\frac{S_{1}\rho _{0}}{b+\lambda } & \ddots & \ddots &  & 0 \\ 
0 & \ddots & \ddots & \ddots & \vdots \\ 
\vdots &  & \ddots & \ddots & 0 \\ 
0 & \cdots & 0 & \rho _{n-2}-\frac{S_{n-1}\rho _{n-2}}{b+\lambda } & 
~S_{n-1}-\sigma _{n-1}-\lambda -S_{n-1}\frac{(T_{n-1}-\sigma _{n-1}+b)}{
b+\lambda }%
\end{matrix}
\right)
\end{align*}
Laplace's theorem applied to the last column of the matrix yields for $%
\lambda \in \mathbb{C}\backslash \{-b\}$ 
\begin{equation*}
Q(\lambda )=(-1)^{n}\left( \prod\limits_{j=0}^{n-1}\left( (S_{j}-\sigma
_{j}-\lambda )(b+\lambda )-S_{j}(-(a_{j}+f_{j})+b)\right)
+(-1)^{n+1}\prod\limits_{j=0}^{n-1}\left( \rho _{j}(b+\lambda
-S_{j+1})\right) \right)
\end{equation*}
Summarizing, we have obtained an expression $Q(\lambda )$ that is equal to
the characteristic polynomial $P(\lambda )$ of $J(S,T)$ for all $\lambda \in 
\mathbb{C} \backslash \{-b\}.$ If we drop the requirement $\lambda \in 
\mathbb{C} \backslash \{-b\},$ $Q$ is a polynomial on $\mathbb{C} $ and thus
a holomorphic function on $\mathbb{C}.$ As a consequence, we have two
holomorphic functions $Q: \mathbb{C}\rightarrow \mathbb{C}$ and $P: \mathbb{C%
} \rightarrow \mathbb{C} $ that are identical on the connected open subset $%
\mathbb{C} \backslash \{-b\}.$ By the well-known identity theorem for
holomorphic functions, $Q=P$ must hold. 
\end{proof}

\bibliographystyle{plain}
\bibliography{cyclicPathogen}

\begin{thebibliography}{10}

\bibitem{OnRouthsCriterion}
S.~Barnett and D.~D. {\v{S}}iljak.
\newblock Routh's algorithm: a centennial survey.
\newblock {\em SIAM Rev.}, 19(3):472--489, 1977.

\bibitem{TomG}
Thomas Greenough.
\newblock personal communication, 2009.

\bibitem{MalariaReview}
Brian~M Greenwood, Kalifa Bojang, Christopher~JM Whitty, and Geoffrey~AT
  Targett.
\newblock Malaria.
\newblock {\em The Lancet}, 365(9469):1487--1498, 2005.

\bibitem{hadinotoShedding}
Vey Hadinoto, Michael Shapiro, Chia~Chi sun, and David~A. Thorley-Lawson.
\newblock The dynamics of ebv shedding implicate a central role for epithelial
  cells in amplifying viral output.
\newblock {\em PLoS Pathog.}, 5(7), 2009.

\bibitem{heffernan}
J.~M. Heffernan, R.~J. Smith, and L.~M. Wahl.
\newblock Perspectives on the basic reproductive ratio.
\newblock {\em J. R. Soc. Interface}, 2:281--293, 2005.

\bibitem{henle}
W.~Henle and G~Henle.
\newblock {\em The Epstien-Barr Virus}, chapter Seroepidemiology of the virus,
  pages 61--78.
\newblock Springer-Verlag, 1979.

\bibitem{hirschSmale}
Morris~W. Hirsch and Stephen Smale.
\newblock {\em Differential equations, dynamical systems, and linear algebra}.
\newblock Academic Press, New York,, 1974.

\bibitem{EBVtransmission}
R.~J. Hoagland.
\newblock The transmission of infectious mononucleosis.
\newblock {\em American Journal of Medical Science}, 229:262--272, 1955.

\bibitem{[15]}
D.~Hochberg, J.~M. Middeldorp, M.~Catalina, J.~L. Sullivan, K.~Luzuriaga, and
  D.~A. Thorley-Lawson.
\newblock Demonstration of the burkitt's lymphoma epstein-barr virus phenotype
  in dividing latently infected memory cells in vivo.
\newblock {\em Proc Natl Acad Sci U S A}, 101(1):239--44, 2004.

\bibitem{[17]}
D.~Hochberg, T.~Souza, M.~Catalina, J.~L. Sullivan, K.~Luzuriaga, and D.~A.
  Thorley-Lawson.
\newblock Acute infection with epstein-barr virus targets and overwhelms the
  peripheral memory b-cell compartment with resting, latently infected cells.
\newblock {\em J Virol}, 78(10):5194--204, 2004.

\bibitem{[18]}
G.~Khan, E.~M. Miyashita, B.~Yang, G.~J. Babcock, and D.~A. Thorley-Lawson.
\newblock Is ebv persistence in vivo a model for b cell homeostasis?
\newblock {\em Immunity}, 5(2):173--9, 1996.

\bibitem{[16]}
L.~L. Laichalk and D.~A. Thorley-Lawson.
\newblock Terminal differentiation into plasma cells initiates the replicative
  cycle of epstein-barr virus in vivo.
\newblock {\em J Virol}, 79(2):1296--307, 2005.

\bibitem{Lin'sPaperOnRobustness}
S.~H. Lin, I.~K. Fong, Y.~T. Juang, T.~S. Kuo, and C.~F. Hsu.
\newblock Stability of perturbed polynomials based on the argument principle
  and {Nyquist} criterion.
\newblock {\em International Journal of Control}, 50(1):55--63, 1989.

\bibitem{controlTheoryWebsite}
John~M. McNamee.
\newblock Mcnamee's bibliography on roots of polynomials, 2002.

\bibitem{NowakEmail}
Martin~A. Nowak.
\newblock personal communication, 2009.

\bibitem{NowakMay}
Martin~A. Nowak and Robert~M. May.
\newblock {\em Virus dynamics}.
\newblock Oxford University Press, Oxford, 2000.
\newblock Mathematical principles of immunology and virology.

\bibitem{nowakMultipleEpitopes}
Martin~A. Nowak, Robert~M. May, and Karl Sigmund.
\newblock Immune responses against multiple epitopes.
\newblock {\em J. Theor. Biol.}, 175:325--353, 1995.

\bibitem{Perelson1999}
A.S. Perelson and P.W. Nelson.
\newblock Mathematical analysis of {HIV-1} dynamics in vivo.
\newblock {\em SIAM Review}, 41(1):3--44, 1999.

\bibitem{GlobaAnOfVirDynamics}
J.~Pr\"uss, R.~Zacher, and R.~Schnaubelt.
\newblock Global asymptotic stability of equilibria in models for virus
  dynamics.
\newblock {\em Math. Model. Nat. Phenom.}, 3(7):126--142, 2008.

\bibitem{BasicEBV_Virology}
A.~B. Rickinson and E.~Kieff.
\newblock Epstein-barr virus.
\newblock In D.M. Knipe and P.M. Howley, editors, {\em Virology}, volume~2,
  pages 2575--2628. Lippincott Williams and Wilkins, New York, 4th ed. edition,
  2001.

\bibitem{souza05}
T.~A. Souza, B.~C. Stollar, J.~L. Sullivan, K.~Luzuriaga, and D.~A.
  Thorley-Lawson.
\newblock Peripheral b cells latently infected with epstein-barr virus display
  molecular hallmarks of classical antigen-selected memory b cell.
\newblock {\em Proc. Natl. Acad. Sci.}, pages 18093--18098, 2005.

\bibitem{EBVreview}
D.A. Thorley-Lawson.
\newblock Epstein-barr virus: exploiting the immune system.
\newblock {\em Nature Reviews Immunology}, 1(1):75--82, 2001.

\bibitem{ThorleyLawson2008195}
David~A. Thorley-Lawson, Karen~A. Duca, and Michael Shapiro.
\newblock Epstein-barr virus: a paradigm for persistent infection - for real
  and in virtual reality.
\newblock {\em Trends in Immunology}, 29(4):195 -- 201, 2008.

\bibitem{tl2004}
David~A. Thorley-Lawson and A.~Gross.
\newblock Persistence of the {Epstein-Barr} virus and the origins of associated
  lymphomas.
\newblock {\em New England Journal of Medicine}, 350(13):1328--1337, 2004.

\bibitem{ChagasReview}
K.~M. Tyler and D.~M. Engman.
\newblock The life cycle of trypanosoma cruzi revisited.
\newblock {\em International Journal for Parasitology}, 31(5-6):472 -- 480,
  2001.

\bibitem{KillerCellDynamics}
Dominik Wodarz.
\newblock {\em Killer Cell Dynamics: Mathematical and Computational Approaches
  to Immunology}.
\newblock Springer Verlag, 2007.

\end{thebibliography}

\end{document}